[1994/12/01]
\documentclass{article}
\title{Posner computing: a quantum neural network model in Quipper}

\usepackage{amsmath,amsthm}
\usepackage{braket}
\usepackage{amsfonts}
\usepackage[hyphens]{url}
\usepackage{graphicx, caption}
\usepackage[margin=1.5in]{geometry}
\usepackage{listings}

\newcommand{\field}[1]{\field{#1}}

\chardef\bslash=`\\ 

\newtheorem{thm}{Theorem}[section]

\theoremstyle{definition}
\newtheorem{defn}{Definition}[section]

\theoremstyle{remark}

\newcommand{\eval}[2][\right]{\relax
  \ifx#1\right\relax \left.\fi#2#1\rvert}

\begin{document}
\author{James L. Ulrich\thanks{Staff scientist, CyberPoint International LLC; email: jlu5@caa.columbia.edu.} }

\markboth{Posner computing: a quantum neural network model in Quipper}
{Posner computing: a quantum neural network model in Quipper}
\renewcommand{\sectionmark}[1]{}
\maketitle

\abstract{We present a construction, rendered in Quipper,  of a quantum algorithm which probabilistically computes a classical function from $n$ bits to $n$ bits. The construction is intended to be of interest primarily for the features of Quipper it highlights.  However, intrigued by the utility of quantum information processing in the context of neural networks, we present the algorithm as a simplest example of a particular quantum neural network which we first define. As the definition is inspired by recent work of Fisher concerning possible quantum substrates to cognition, we precede it with a short description of that work.}

\section*{Acknowledgement} The author wishes to thank Neil J. Ross for helpful discussions and code contributions, and Paul Li, Mark Raugas, and Riva Borbely for helpful reviewing.

\section{Introduction}
We seek in this work to add to the corpus of quantum algorithms expressed in the Quipper quantum programming language \cite{Green},  by presenting a construction of a quantum algorithm which probabilistically computes a classical function from $n$ bits to $n$ bits. Of course, since a general quantum circuit can compute exactly any algorithm constructed using classical gates \cite{Nielsen}, the algorithm is intended to be of interested primarily for the features of Quipper it highlights, and for its restriction to particular types of gates. The code itself is presented in the appendix.  However, intrigued by the utility of quantum information processing in the context of neural networks, we present it as a simplest construction of a particular quantum neural network which we define below.  

The definition, in turn, is inspired by an article by Fisher \cite{Fisher}, in which he conjectures that the phosphate ion's half-integer spin may serve as the brain's ``qubit'' (i.e., unit of quantum storage), and that pairs of such ions form spin singlet states, which are preserved inside ``Posner molecules''. These cube-shaped molecules inherit a tri-level ``pseudo-spin'' from the six phosphate ions they contain, characterizing spin eigenstate transformations under rotations along the cube diagonal. Posner molecules may bond in pairs, collapsing onto a zero total pseudo-spin state leading to release of calcium, which in turn enhances neuronal firing. Posner molecules in different neurons are posited to become entangled, producing cross-neuronal firing correlations which are quantum in origin. 

Hence, we consider a model for a feed-forward, layered neural network inspired by the described interactions of Posner molecules. The network resembles a traditional feed-forward neural network, except that in any given layer, the activation functions of one or more of the units are replaced with a single quantum operation, using a specified (not necessarily universal) gate set, and taking the units' classical inputs to their classical outputs.   The algorithm presented below may be seen as a simplest example of such a network.

\section{Related work}
From a philosophical perspective, there is a grand tradition of attempting to model cognition as an artifact of quantum information processing in the brain. In his 1989 book ``The Emperor's New Mind,'' Penrose introduced the idea that consciousness is an artifact of the gravity-induced collapse of a quantum-mechanical wave-function governing brain states, referred to as ``Orchestrated Objective Reduction (Orch-Or) '' \cite{Penrose1}. A few years later, Albert \cite{Albert} explored the non-intuitive consequences for mental belief states of the application of the Copenhagen interpretation of quantum mechanics to human observers \cite{Albert}, suggesting the standard interpretation is wrong, thereby supporting Penrose. In the following decade, work by Khrennikov \cite{Khrennikov} posited that the brain is a ``quantum-like'' computer, in that one may observe interference patterns in statistical descriptions of mental states.   An updated version of the Penrose's ``Orch-Or'' hypothesis is provided in \cite{Penrose2}, and Fisher's recent article \cite{Fisher} might reasonably be viewed as providing an alternative physical substrate for this hypothesis. A key element of this substrate is that while the processes describing neuronal inputs and outputs may be  classical, groups of neurons may coordinate firing through quantum processes.  It is this element which motivates the description of the neural network described below, which we dub a ``Posner'' network, in honor of its origin. 

In terms of extending classical neural networks to the quantum world, approaches may be found in (for example) \cite{Ricks} and \cite{Behrman}, and a survey may be found in \cite{Garman}. Quantum perceptron networks and/or perceptron networks using quantum computation in the training phase are given in (for example) \cite{Wiebe}, \cite{Zhou}, and \cite{Zidan}, and a review of some approaches may be found in \cite{Schuld}. A general framework for quantum machine learning is presented in \cite{QELM}.  In the model presented here, quantum states exist only within a network layer; the layer inputs and outputs are purely classical. Moreover, we stress that our primary aim is simply to contribute to the corpus of extant Quipper code.

\section{Posner molecule interactions}

For the interested reader, we here provide a brief summary of the aspects of \cite{Fisher} motivating the present work, and encourage the reader to consult that source for more detail. The remainder of this paper does not depend on this section, however.

The mechanism for quantum cognition described in Fisher's work proceeds more or less as follows. It starts with two phosphate ions with nuclear spin $\frac{1}{2}$ held inside a magnesium shell in extracellular fluid, and emitted in a spin singlet state $\frac{1}{2}(|0 \rangle - | 1 \rangle)$. The entangled ions are then drawn into (possibly distinct) presynaptic glutematergic neurons, where they participate in the formation of Posner molecules, each  containing six phosphorous ions. The ions, when viewed along the diagonal 3-fold symmetry axis of the molecule, form a hexagonal ring, and the associated spin Hamiltonian is given by $H = \sum_{ij} J_{ij} \vec{S}_i \cdot \vec{S}_j$ where $i,j=1,2, \cdots, 6$ label the $6$ spins, the $\vec{S}_l$ are the spin operators, and the $J_{ij}$ are exchange interactions. The $2^6$ energy eigenstates of $H$ can be labelled by their transformation properties under a 3-fold rotation, acquiring a phase factor  $e^{i \sigma 2\pi/3}$ with $\sigma = 0, +1,-1.$ Hence the molecule can be described by pseudo-spin states $|0\rangle, | 1 \rangle, \ket{-1}$, with overall molecular state of the form $| \Psi_{Pos} \rangle = \sum_\sigma c_\sigma | \psi_\sigma(\phi) \rangle | \sigma \rangle$, where $\psi_\sigma$ encodes rotational spin state for angle $\phi$ about the diagonal, given the state $\ket{\sigma}$. \cite{Fisher2}.

The state of two Posner molecules  $(a, a^\prime)$ occupying a neuron is given by
\begin{displaymath}
| \Psi_{a a^\prime} \rangle = \sum_{\sigma, \sigma^\prime} C^{a a^\prime}_{\sigma \sigma^\prime} |\psi^a_\sigma \rangle | \psi^{a^\prime}_{\sigma^\prime}\rangle |\sigma \sigma^\prime\rangle_{a a^\prime},\,\,\,  \sum_{\sigma, \sigma^\prime} |C^{a a^\prime}_{\sigma \sigma^\prime}|^2 = 1.
 \end{displaymath}
 which is entangled unless $C^{a a^\prime}_{\sigma \sigma^\prime} =  c^a_\sigma   c^{a^\prime}_{\sigma^\prime}$, and where by $c^a_\sigma$ (respectively, $c^{a^\prime}_{\sigma^\prime})$, we mean the coefficient of the $\sigma$ component of the pseudo-spin of molecule $a$ (respectively, $\sigma^\prime$ component of $a^\prime$).

 Chemical bonding of the molecules is equivalent to collapse of $| \Psi_{a a^\prime} \rangle$ onto a total pseudo-spin $0$ state. In this case both molecules melt, releasing calcium, which in turn stimulates further glutamate release into the synaptic cleft, impacting the firing behavior of the \emph{post}-synaptic neuron. The probability of bonding is given by 
\begin{equation} \label{eq:bond}
P^{a a^\prime}_{\text{react}} = \sum_{\sigma, \sigma^\prime} |C^{a a^\prime}_{\sigma \sigma^\prime}|^2 \delta_{\sigma + \sigma^\prime, 0},
\end{equation}
where $\delta$ denotes the Kronecker delta function.

Now let $| \Psi_{a a^\prime} \rangle \otimes | \Psi_{b b^\prime} \rangle$ be the state encoding two pairs of Posner molecules $(a,a^\prime)$ and $(b, b^\prime)$ with $(a,b)$ in neuron 1, and $(a^\prime, b^\prime)$ in neuron 2, with $(a,a^\prime)$ entangled,  as are  $(b,b^\prime)$.  Let $r = 1$ if $a,b$ bind and  $0$ otherwise, and similarly let $r^\prime = 1$ if $a^\prime, b^\prime$ bind and $0$ otherwise. Then the joint probability of a given combination $r, r^\prime$ is given by 
\begin{equation}\label{eq:entangled}
P_{r r^\prime} = \sum_{\sigma_a \sigma_{a^\prime} }\sum_{\sigma_b \sigma_{b^\prime}} |C^{a a^\prime}_{\sigma_a \sigma_a^\prime}|^2 |C^{b b^\prime}_{\sigma_b \sigma^\prime_b}|^2 g_r(\sigma_a, \sigma_b) g_{r^\prime}(\sigma_{a^\prime}, \sigma_{b^\prime})
\end{equation}
where $g_1(\sigma, \sigma^\prime) = \delta_{\sigma + \sigma^\prime,0}$ and $g_0 = 1 - g_1$. 

Fisher defines an ``entanglement measure''  $\mathcal{E} = [\delta r \delta r^\prime]$ where $\delta r := r - [r]$ and $[f_{r r^\prime}] := \sum{r r^\prime}P_{r r^\prime}f_{r r^\prime}$.  If $\mathcal{E} > 0$ then the two binding reactions themselves are positively correlated by virtue of quantum entanglement, and if $\mathcal{E} < 0$ then they are anti-correlated.  We seek to capture this feature of interneuronal entanglement in our definition of a ``Posner'' network, which follows in the next section.

\section{A quantum neural network for Posner computing}

We now define a quantum neural network inspired by the elements of the preceding section. For convenience and to fix notation,  we recall that a feed-forward, hard-threshold perceptron network \cite{Russell}, may be thought of as a map from $\mathbb{Z}_2^n\rightarrow \mathbb{Z}_2^m$, given by a direct acyclic graph $G$ with vertices $v_{ij}$,  $i = 1, \cdots, n_{L_j}$, organized into layers $L_j$, $j = 1, \cdots k$, with edges $e_{ij, i^\prime (j+1)}$ running only between layers $L_j$ and $L_{j+1}$. For every layer $L_j$ aside from the first, each $v_{i^\prime j}$ takes a value given by a binary-valued function $g_{i^\prime j} (\sum w_{i (j-1),i^\prime j} v_{i (j-1)}  + d_i^\prime)$, where the sum is taken over all $v_i$ in layer $L_{j-1}$ such that there is an edge $e_{i(j-1),i^\prime j}$ from $v_{i (j-1)}$ to $v_{i^\prime j}$, and the $w_{i (j-1),i^\prime j}$ are weights assigned to the respective edges, and $d_i$ is a constant. To evaluate the map on a binary input $b_1,b_2 , \cdots, b_n$, we set the $v_{i1} = b_i$. Iteratively evaluating each successive layer in turn, given the inputs provided by the preceding layer, the corresponding output will then be given by the values of the units $v_{1 k}, \cdots , v_{m k}$ in the last layer $L_k$.  To obtain our quantum neural network, we modify this classical neural network in the following way:

\defn \label{def:SPN} A \emph{Simple Posner $\mathcal{R}$-Network} (SPN-$\mathcal{R}$)  is a feed-forward, hard-threshold perceptron network as described above, with the following amendations. For at least some units $v_{i_1 j}, \cdots, v_{i_p j}$  in at least one layer $L_j$:

\begin{enumerate}
\item  The functions $g_{ij}$ are replaced with a single quantum circuit \cite{Nielsen} on $p$ qubits, consisting of a single unitary transformation $U = U_{(i_1, \cdots, i_p), j}: \mathbb{C}^{2^{ \otimes_p}} \rightarrow \mathbb{C}^{2^{\otimes_p}}$, followed by a projective measurement onto a computational basis state. The possible values of $v_{i_1 j}, \cdots, v_{i_p j}$ form the $2^p$ computational basis vectors $b_q$ of $\mathbb{C}^{2^{\otimes_p}}$. The circuit $U$ must be constructed from a specific (not necessarily universal) gate set $\mathcal{R}$. 

\item Evaluating $U$ on a given basis vector $b$ produces a vector $v = \sum_{q = 1}^{2^p} c_q b_q$. Randomly projecting onto a basis vector $b_q$,  where $b_q$ is chosen with probability $|c_q|^2$, determines the outputs of the $v_{i_1 j}, \cdots, v_{i_p j}$, where the $r$-th bit of $b_q$ is the value assigned to $v_{r j}$. 

\item Each $v_{i j} \in \{ v_{i_1 j}, \cdots, v_{i_p j} \}$ is connected to exactly one unit $v_{i^\prime (j-1)}$ in layer $L_{j-1}$, and the connecting edge has unit weight. The basis vector on which $U$ is evaluated is determined by the values of the $v_{i^\prime (j-1)}$.
\end{enumerate}

An example of an SPN-$\mathcal{R}$ appears in figure \ref{fig:SPN1}. \\
\begin{figure}
  \centering
  \includegraphics[width=.7 \linewidth]{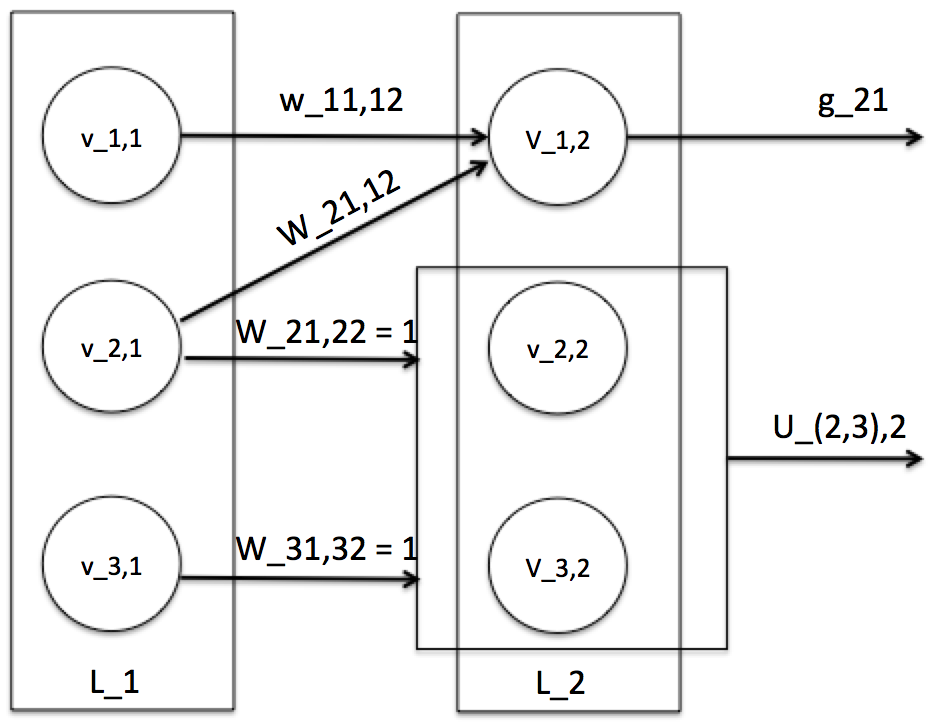}
  \captionsetup{width=.7\linewidth}
  \caption{\small{A two-layer SPN-$\mathcal{R}$. The outputs of units two and three in layer 2 are determined by a unitary transformation on a state corresponding to the values of units two and three in layer one, followed by a projective measurement.}}
  \label{fig:SPN1}
\end{figure}

It follows from definition \ref{def:SPN} that a quantum circuit implementing a given unitary transformation $\mathcal{U}$ on a quantum register initialized from $n$ classical bits, followed by a projective measurement onto an $n$-bit basis state from which an $n$-bit classical output is read, may be viewed as a two-layer SPN-$\mathcal{R}$. More (perhaps) interesting constructions, in terms of the Quipper implementations to which they lead, arise from restricting the gate set $\mathcal{R}$, or combining such two-layer SPN-$\mathcal{R}$ . Clearly there are SPN-$\mathcal{R}$s that cannot be represented by a single unitary transformation: for example, the three-layer SPN-$\mathcal{R}$ coupling two copies of a single-qubit quantum circuit performing a Hadamard transformation followed by a projective measurement in the $\ket{0}, \ket{1}$ basis (see figure \ref{fig:SPN2}).

\begin{figure}
  \centering
  \includegraphics[width=.7 \linewidth]{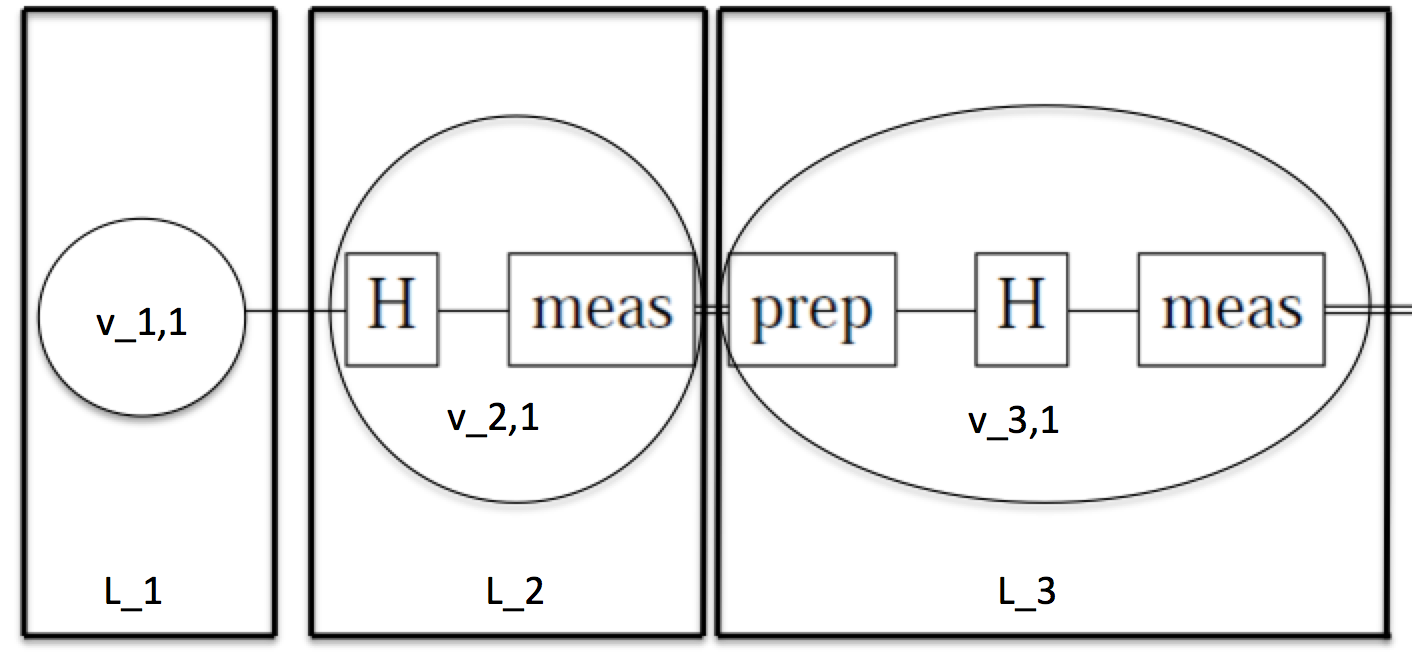}
  \captionsetup{width=.7\linewidth}
  \caption{\small{A three-layer SPN-$\mathcal{R}$ consisting of a Hadamard transformation (in $L_2$) on a singe qubit initialized to the output of layer $L_1$, followed a projective measurement, output of which passes to $L_3$, where another Hadamard transform is performed, followed by another measurement. The SPN-$\mathcal{R}$ does not correspond to a single unitary transformation, because of the projective measurement in layer $L_2$.}}
  \label{fig:SPN2}
\end{figure}

\subsection{SPN-$\mathcal{R}$s can learn arbitrary classical functions} 

To illustrate the capability (and limitations) of the simplest SPN-$\mathcal{R}$ configuration, and to motivate the construction of our Quipper circuit, we present the following theorem.

\begin{thm} \label{thm:SPN} Let $n$ be an integer such that there exists an angle $\theta \in (0, \frac{\pi}{2} )$ such that 
\begin{equation}\label{equ:theta}
\frac{1}{\sqrt{2}^n} \big( \text{cos}^{2^n-1}(\theta) + \sum_{i=0}^{2^n-2}\text{cos}^i(\theta)\text{sin}(\theta) \big) > \frac{1}{\sqrt{2}}. 
\end{equation}

Let $\mathcal{R}$ be the quantum gate set consisting of the Hadamard gate, swap gate, and controlled, two-level gates which act on basis states $\ket{0}, \ket{i}$ and are built from $\left( \begin{array}{cc} cos(\theta) & sin(\theta)  \\ sin(\theta) & -cos(\theta) \end{array} \right)$. Then for any surjective function $f: \mathbb{Z}_2^n \rightarrow \mathbb{Z}_2^n$ taking $n$-bit binary strings to $n$-bit binary strings, there is a two-layer $SPN-\mathcal{R}$ which can ``learn'' $f$, in that given an input $b$, it computes $f(b)$ with probability greater than $\frac{1}{2}$.

\end{thm}

\begin{proof}

We construct the SPN-$\mathcal{R}$  in following way. A surjective function $f$ with range $\mathbb{Z}_2^n$ takes $2^n$ distinct values, corresponding to the basis vectors $\ket{i}$ of $\mathbb{C}^{2^{\otimes_n}}$. Taking $\ket{0} \in  \mathbb{C}^{2^{\otimes_n}}$ and applying a Hadamard transform yields an even superposition $\ket{\psi}$ of all basis vectors. To bias this state towards a given basis vector $\ket{k}$ corresponding to $f = k$, one can apply, for each $j$,  a $y$-axis rotation of fixed angle $\theta$ towards $\ket{k}$ in the Bloch sphere corresponding to  $\{ \ket{k}, \ket{j} \}$. For a projective measurement on $\ket{\psi}$ to select $\ket{k}$ with odds greater than $\frac{1}{2}$, $\theta$ must satisfy (\ref{equ:theta}); the left-hand side of the inequality represents the overall transformation of the co-efficient of $\ket{k}$ within the basis expansion of $\ket{\psi}$.

Couched in terms of ``learning'', we take the training data to be the set of $2^n$ pairs of $n$-digit binary strings $\{b_{\text{input}}, b_{\text{output}} \}$ defining $f$, where  $b_{\text{input}}$ is an input to $f$ and $b_{\text{output}}$ is the associated output. The SPN-$\mathcal{R}$ takes the form of a quantum circuit accepting $n$ classical bits, on which it performs a single (multi-gate) unitary operation, followed by a projective measurement onto the computational basis. The quantum circuit consists of $n$ quantum input bits (qubits) and $n$ quantum ancilla qubits. The circuit first initializes $n$ ancilla qbits  to the $\ket{0}$ state,  and then performs a Hadamard transformation to take $\ket{0} \rightarrow \psi_1 = \frac{1}{\sqrt{2}^n} \sum_{i=0}^{2^n-1} \ket{i}$. For $U =  \left( \begin{array}{cc} cos(\theta) & sin(\theta)  \\ sin(\theta) & -cos(\theta) \end{array} \right) $, and each $0 < i \leq 2^n-1$, a two-level unitary transformation $U_{0,i}$ is performed on the space spanned by the basis states $\ket{0} \text{and}  \ket{i}$ ( $U_{0,i}$ acts trivially on all other subspaces). Collectively, the $U_{0,i}$ take $\psi_1$ to a state $\psi_2$ in which the $\ket{0}$ component has amplitude $\frac{1}{\sqrt{2}^n} \big( \text{cos}^{2^n-1}(\theta) + \sum_{i=0}^{2^n-2}\text{cos}^i(\theta)\text{sin}(\theta) \big)$. Then, for each $b_{\text{output}}$, the circuit performs a controlled, two-level $X$ gate swapping $\ket{0}$ and $\ket{b_{\text{output}}}$, the latter being the basis state whose $i$-th qubit corresponds to the value of the $i$-th bit of $b_{\text{output}}$, and where the controls are given by $b_{\text{input}}$ and placed on the input qubits. Hence the amplitude of $\ket{0}$ and  $\ket{b_{\text{output}}}$ are swapped precisely when the input qubits are initialized according to $b_{\text{input}}$.

\end{proof}

A Quipper-based implementation of the above algorithm may be found at \cite{Ulrich} and is reproduced in the appendix. The implementation uses  $\theta = \frac{\pi}{2^{(n/2)+1}}$ for even $n$, and  $\theta = \frac{\pi}{2^{((n-1)/2)+1}}$ for odd $n$. These choices of $\theta$ have been computer-verified to satisfy (\ref{equ:theta}) for $1 < n \leq 26$. Figure \ref{fig:SPN_comp_circuit}  depicts a Quipper circuit printout for the ``complement'' function $f :  (b_1, b_2) \rightarrow (b_1 \oplus 1, b_2 \oplus 1)$.  Note that in the diagram, $U$, a $y$-rotation,  is composed as a basis change followed by a  $z$-rotation, to accommodate the Quipper gate set. 

\begin{figure}
  \centering
  \includegraphics[width=.85 \linewidth]{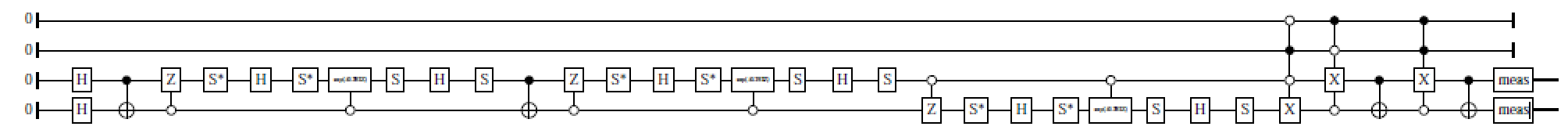}
  \captionsetup{width=.85 \linewidth}
  \caption{\small{Quantum circuit computing the ``complement'' function for the two qubit case.}}
  \label{fig:SPN_comp_circuit}
\end{figure}

\section{Summary and avenues for further research}

With the goal of increasing the corpus of extant Quipper programs, and inspired by recent work on molecular substrates for quantum effects having cognitive impact, we have suggested a model for a ``Posner'' neural network, which is a perceptron network in which the activation functions of some individual units are replaced with a quantum circuit, expressed in a specific gate set, and  implementing a single unitary transformation, followed by a projective measurement. The simplest form of such a network is simply an individual (restricted) quantum circuit itself.  Further, we have presented a quantum algorithm, expressed in the Quipper programming language, by which a particular instance of such a neural network can probabilistically compute any function mapping $n$-bit inputs to $n$-bit outputs.  Whether more interesting Posner networks provide additional capabilities over the most trivial instances, or their classical counterparts, remains an open question for the author.

\section{Appendix}
We present here the Quipper source code implementing the circuit described in the text.
\begin{lstlisting}[language=Haskell]

{-


author: James Ulrich julrich@cyberpointllc.com


Copyright (c) 2016 CyberPoint International LLC


Permission is hereby granted, free of charge, to any person obtaining a copy of
this software and associated documentation files (the "Software"), to deal in 
the Software without restriction, including without limitation the rights to 
use, copy, modify, merge, publish, distribute, sublicense, and/or sell copies of
the Software, and to permit persons to whom the Software is furnished to do so, 
subject to the following conditions:

The above copyright notice and this permission notice shall be included in all 
copies or substantial portions of the Software.

THE SOFTWARE IS PROVIDED "AS IS", WITHOUT WARRANTY OF ANY KIND, EXPRESS OR 
IMPLIED, INCLUDING BUT NOT LIMITED TO THE WARRANTIES OF MERCHANTABILITY, FITNESS 
FOR A PARTICULAR PURPOSE AND NONINFRINGEMENT. IN NO EVENT SHALL THE AUTHORS OR 
COPYRIGHT HOLDERS BE LIABLE FOR ANY CLAIM, DAMAGES OR OTHER LIABILITY, WHETHER
IN AN ACTION OF CONTRACT, TORT OR OTHERWISE, ARISING FROM, OUT OF OR IN 
CONNECTION WITH THE SOFTWARE OR THE USE OR OTHER DEALINGS IN THE SOFTWARE.
-}

import Quipper
import Data.List.Split
import QuipperLib.Synthesis
import Quantum.Synthesis.Ring
import Quantum.Synthesis.Matrix
import Quantum.Synthesis.MultiQubitSynthesis
import Debug.Trace
import QuipperLib.Simulation
import System.Random
import Data.Map (Map)
import qualified Data.Map as Map
import Quipper.QData

{-

This program will take as input a 1-1 map of key-value pairs, where the key is 
an n-digit binary number (an "input"), and the value is an n-digit binary 
number (the corresponding "output"). It will produce as output a quantum 
circuit, which with high probability, maps a given input to the corresponding 
output.  

We assume program is invoked as 'posner2 < func.txt' where func.txt is a file 
of the  form (here using example n=2):
        
    00:00   
    01:01
    10:10
    11:11
       01
   <mode>
                   
the first 2^n lines define a bijective function from bitstrings of length n to 
bitstrings of length n; the digits to the left of the ':' are the inputs. the 
lines must appear in increasing order of  outputs, viewed as binary strings 
(e.g. 00, 01, 10, 11, for n=4). the next to last line 
is a specific input to be evaluated by simulating the circuit (used for mode "S"
only). the last line is a run mode: 
    
    A=show amplitudes via sim_amps (supported for n <= 2 only), for input |0>.
    P=print_generic.
    S=simulate via run_generic.


-- major algorithmic steps (in concept, if not actual code) are:

1. ingest a string giving the map of (input,output) pairs and 
   store in a suitable data structure in memory.
   
2. define an n-quibit circuit that takes the |0> state to 
   one that is a superposition S of all possible n-qubit basis states, 
   but for which the state |0> has the largest probability amplitude. We do
   by first applying a hadamard gate, and then applying the two level "bias" 
   gate  [[cos(pi/f(n)), sin(pi/f(n))], [sin(pi/f(n)), -cos(pi/f(n))]] 
   successively to states (|0>,|1>), (|0>,|2>), ...., (|0>,|2^(n-1)>),
   for suitable f(n).
        
   
3. add n ancilla qubits to the circuit. then for each output i given by the 
   map, add a controlled swap gate exchanging the weight of the |0> component 
   of S with that of the component corresponding to |i>, with controls given 
   by the input producing |i>. the controls are placed on the ancillas.
   
4. to evaluate the circuit for a given input, the ancillas are initialized to
   the input.
    
-}

-- declare a map data type entry
data InputOutput = InputOutput {
   
   input :: [Bool],
   output :: [Bool]
} deriving (Show)

-- declare a function truth table (input->output map defining a function)
data TruthTable = TruthTable {
    entries :: [InputOutput]
} deriving (Show)
 

-- produce a quantum circuit from the function defined by the input 
processInput :: [String] -> ([Qubit] -> Circ [Qubit])
processInput table = do

    -- from input lines defining a function (as a text-based truth table of 
    -- inputs->outputs), get list of input,output string pairs
    let u = map (\x -> splitOn ":" x) table
    
    -- from list of input, output string pairs, get list of 
    -- input, output boolean list pairs, make internal truth table from that
    let v2 = map (\x -> (str_to_blist(x !! 0),str_to_blist(x !! 1)) ) u  
    let g2 = map (\x -> InputOutput (fst x) (snd x)) v2      
    let tt = TruthTable g2   

    -- make quantum circuit based on table 
    let circ = tt_circuit tt
    circ
    

-- convert a text string of 0s and 1s to a boolean array   
str_to_blist :: String -> [Bool]
str_to_blist lc = 
    let bl = map (\x -> case x of '1' -> True
                                  '0' -> False) lc
    in bl                             
                            

--  A basis change to obtain y-rotations from z-rotations. 
--  courtesy Neil J. Ross.                    
y_to_Z :: Qubit -> Circ Qubit
y_to_Z q = do
  q <- gate_S_inv q
  q <- gate_H q
  q <- gate_S_inv q
  return q


{-
the "bias" gate R_Y(pi/f(n))*Z; see Nielsen and Chuang, Quantum Computation
and Quantum Information, pp.175-176. Decomposition in terms of Quipper gates 
courtesy Neil J. Ross.
-}
my_biasUt :: Timestep -> Qubit -> Circ ()
my_biasUt t q = do
  
    -- q <- expZt pi q
    q <- gate_Z q
    q <- with_computed (y_to_Z q) (expZt t)
    return ()


apply_controlled_X :: ControlList -> Qubit -> Circ()
apply_controlled_X controls q = do
    q <- gate_X q `controlled` controls 
    return ()
    
-- repeatedly apply our two-level bias gate  as 
-- a U_{0,i} gate for i, i-1, i-2, ..., 1
bias_gate_recurse  :: (Qubit-> Circ()) -> Int -> [Qubit] -> Circ()
bias_gate_recurse g 0 qbits =  return()
bias_gate_recurse g i qbits = do 
    twolevel 0 i qbits g
    bias_gate_recurse g (i-1) qbits  

-- construct a control list for a given set of function inputs, expressed 
-- as bools
control_recurse :: ControlList -> Int -> [Qubit] -> [Bool] -> ControlList
control_recurse controls idx qubits bools = 
    case (idx == 0) of 
        True -> controls
        False -> do 
            let new_controls = controls .&&. (qubits !! idx) .==. (bools !! idx)    
            control_recurse new_controls (idx-1) qubits bools

-- apply a controlled X gate as U_{0,idx} with controls on inputs corresponding
-- to output state |idx> for thruth table entries idx = 1....n                  
flip_gate_recurse ::  TruthTable -> Int -> [Qubit] -> [Qubit] -> Circ()
flip_gate_recurse tt idx ancillas qbits = do

     
    -- get the input booleans for the output given by idx
    let pos_neg = input ( (entries tt) !! idx ) 

    -- now build a control list based on the booleans. we recurse to build it 
    -- up wire by wire using the .&&. operator    
    let init_controls = (ancillas !! 0) .==. (pos_neg !! 0)    
    let controls = control_recurse init_controls (length(pos_neg) -1 ) 
                                    ancillas pos_neg
        
    let g = apply_controlled_X controls
    twolevel 0 idx qbits g

    case (idx == (length (entries tt) - 1)) of 
        True -> return()
        False -> flip_gate_recurse tt (idx + 1) ancillas qbits

       
-- construct a circuit that computes the output for the given input and truth 
-- table
tt_circuit :: TruthTable -> [Qubit] -> Circ [Qubit]
tt_circuit tt ancilla = do

    -- n = number of input bits
    let n = length (input ( (entries tt) !! 0) )
    
    -- create n circuit bits and set them all to 0s                                
    qbits <- qinit (replicate n False)
    
    -- now perform a hadamard on the qbits
    mapUnary hadamard qbits
   
    -- apply our 2-level bias gate to (|0>,|1>), (|0>,|2>), ..., (|0>,|2^(n-1)>)         
    let p = pi/(get_divisor n)
    let g = my_biasUt p  
    
    -- step 2
    bias_gate_recurse g (2^n -1) qbits
    
    -- step 3
    flip_gate_recurse tt 1 ancilla qbits
  
    qdiscard ancilla
    return qbits

get_divisor :: Int -> Double
get_divisor n = case (even n) of 
                    True -> 2^((div n 2) + 1) 
                    False ->  2^((div (n-1) 2) + 1) 
                     
    
-- here follow some utility routines to help with debugging    
tester :: [Bool]  -> ([Qubit] -> Circ [Qubit]) -> Circ [Qubit]
tester ibools circ = do
    ibits <- qinit ibools 
    circ ibits
    

-- a sample input state for a 2-qbit circuit
myMap2 :: Map [Bool] (Cplx Double)
myMap2 = Map.fromList (map makePair
    [[False, False],[False,True], [True,False],[True,True]])
    where makePair x =  case x of [False,False] -> ([True, True],1.0)
                                  [False, True] -> ([False,True], 0.0)
                                  [True, False] -> ([True,False], 0.0)
                                  [True, True]  -> ([False,False],  0.0)
                                                                          
-- and here's main!                                        
main = do

    s <- getContents  
    let t = lines s 
    
    -- separate the lines defining the function to be evaluated from the 
    -- input on which it is to be evaluated, and the mode
    let func_def = take ((length t)- 2) t
    let input = t !!  ((length t) - 2)
    let mode = t !! ((length t) - 1)
    let ibools = str_to_blist input
    let ibits = map (\x -> cinit x) ibools
    
    -- get a quantum circuit that evaluates the function given an input 
    let qa = processInput func_def

    case mode of
        "A" -> do
            -- WE USE THIS CODE TO GET OUTPUT STATE AMPLITUDES FOR THE 
            -- FUNCTION DEFINED BY STDIN, AND THE INPUT HARDCODED IN myMap, 
            -- myMap2  RTNS 
            case ((length ibools) == 2) of 
                True -> do
                    let m = myMap2
                    t <- randomIO
                    putStrLn("here is the input map: " ++  Map.showTree(m))      
                    let res2 = sim_amps (mkStdGen t) qa m     
                    let s3 = Map.showTree(res2)      
                    putStrLn("here is the output map: ")
                    putStrLn(s3)
                False -> do     
                    putStrLn("option A only supported for n=2")

        "P" -> do
            -- WE USE THIS CODE TO PRINT THE CIRCUIT, GIVEN THE FUNCTION
            -- DEFINED BY STDIN
            let circ = tester ibools qa 
            print_generic Preview circ
        
        "S" -> do   
        
            case ((length ibools) >= 2) of 
                True -> do
                
                    -- WE USE THIS CODE TO RUN MANY SIMULATION 
                    -- TRIALS OVER THE CIRCUIT, GIVEN THE 
                    -- FUNCTION AND INPUT DEFINED BY STDIN
                    for 1 100 1 $ \i -> do  
                        t <- randomIO
                        let res = run_generic (mkStdGen t) (1.0::Double) qa ibools
                        putStrLn(show(res))
                    endfor
                    
                False -> do     
                    putStrLn("option S only supported for n>=2")
   


\end{lstlisting}

\end{document}